\documentclass{amsart}[11pt]
\usepackage{graphicx}
\usepackage[T1]{fontenc}
\usepackage[table]{xcolor}
\usepackage{geometry}
\usepackage{booktabs}
\usepackage{array}
\usepackage{caption}
\usepackage{url}
\usepackage[colorlinks=true,citecolor=blue,linkcolor=blue,urlcolor=blue,bookmarks,bookmarksopen,bookmarksdepth=2,backref=page,breaklinks]{hyperref}
\usepackage{float}
\usepackage{enumerate}  
\usepackage{xcolor}
\usepackage{colortbl}
\usepackage{multirow}
\usepackage{multicol}
\usepackage{wrapfig}
\usepackage{amssymb}
\usepackage[T1]{fontenc}
\usepackage[scaled]{beramono} 
\usepackage{array}
\usepackage{seqsplit} 
\usepackage{geometry}

\usepackage{algorithm}
\usepackage{algpseudocode}

\newcolumntype{T}[1]{>{\ttfamily\arraybackslash}p{#1}}
\definecolor{bad}{cmyk}{0, 0.7808, 0.4429, 0.1412}

\DeclareFontFamily{U}{wncy}{}
\DeclareFontShape{U}{wncy}{m}{n}{<->wncyr10}{}
\DeclareSymbolFont{cyrletters}{U}{wncy}{m}{n}
\DeclareMathSymbol{\Sha}{\mathalpha}{cyrletters}{"58}

\long\def\gloops#1{}
\usepackage{environ}
\NewEnviron{noproof}{}

  \usepackage[all]{xy}
\newtheorem{example}{Example}
\newtheorem{invariant}{Invariant}
\newtheorem{fact}{Fact}
\newtheorem{question}{Question}
\newtheorem{proposition}{Proposition}
\newtheorem{remark}{Remark}

\newtheorem{lemma}{Lemma}
\newtheorem{problem}{Problem}
\newtheorem{conjecture}{Conjecture}
\def\binogauss(#1,#2){\genfrac{[}{]}{0pt}{}{#1}{#2}}

\def\field(#1){{\mathbb F}_{#1}}
\def\cczeq(#1,#2){#1\underset{\text{CCZ}}{\sim}#2}
\def\eaeq(#1,#2){#1\underset{\text{EA}}{\sim}#2}

\def\degmin(#1){\deg_{{\rm min}}(#1)}
\def\degmax(#1){\deg_{{\rm max}}(#1)}

\usepackage{xstring}

\def\fd(#1){%
  \IfStrEq{#1}{1}
    {\mathbb{F}_2}
    {\mathbb{F}_2^{#1}}%
}

\def\aglgroup(#1){\textsc{agl}(#1)}
\def\glgroup(#1){\textsc{gl}(#1)}
\def\orbit(#1){\textsc{orbit}(#1)}
\def\stable(#1){\textsc{stab}(#1)}
\def\class(#1){\textsc{class}(#1)}
\def\flat(#1){\mathfrak{F}_{#1}}
\def\ffd(#1){{\mathbb F}_{2^{#1}}}
\def\comp(#1){{\rm comp}(#1)}
\def\compancle(#1){{\langle #1 \rangle}}
\def\abs(#1){\vert #1\vert}
\def\card(#1){\vert #1\vert}

\def\tfr(#1,#2){ { {\widehat{#1}(#2)} } }
\def\TFR(#1,#2,#3){ { {\widehat{#1}(#2,#3)} } }
\def\spec(#1){  {\rm spec}{\,(#1)}}  
\def\moment(#1,#2){ { \kappa_{#1}(#2)} }
\def\triv(#1,#2){ { AROUND{\rm triv}_{#2}(#1)} }

\def\invec{{\mathfrak J}'}

\def\boole(#1){ B(#1) }
\def\vst(#1,#2,#3,#4){ (#1,#2)^{#4}_{#3} }
\def\bst(#1,#2,#3){ B_{#1}^{#2}(#3) }
\def\nl(#1){  {\rm nl}{\,(#1)}} 
\def\lin(#1){  {\rm lin}{\,(#1)}}
\def\l(#1){  {\rm l}{\,(#1)}} 
\def\CCZ{\textsc{ccz }}
\def\CCZ-{\textsc{ccz}-}
\def\APN{\textsc{apn} }
\def\EA{\textsc{ea} }
\def\MNBC{\textsc{mnbc} }
\def\EA-{\textsc{ea}-}
\def\ext(#1){{\rm Ext}(#1)}
\def\bleu#1{\textcolor{black}{#1}}
\def\rouge#1{\textcolor{black}{#1}}

\def\invj{\mathfrak j}
\def\invw{\omega}

\def\invnu(#1){{\mathfrak N}(#1)}
\def\invdl(#1){{\mathfrak L}(#1)}

\def\invcomp(#1,#2){{#1}_{{\rm c}}(#2)}

\def\invnum(#1){\mathfrak n(#1)}
\def\invder(#1,#2){{#1}^{{\rm d}}(#2)}
\def\invres(#1,#2){{#1}^{{\rm r}}(#2)}
\def\invmul(#1,#2){X_{#1,#2}}
\def\endomul(#1,#2,#3){X_{#1,#2}^{#3}}

\def\invval{{\mathfrak V}}
\def\invtrans(#1){{\mathfrak T}_{#1}}

\def\compo(#1){\textsc{C}(#1)}
\def\Der(#1,#2){{\rm der}_{#1}(#2)}

\def\restriction#1#2{\mathchoice
              {\setbox1\hbox{${\displaystyle #1}_{\scriptstyle #2}$}
              \restrictionaux{#1}{#2}}
              {\setbox1\hbox{${\textstyle #1}_{\scriptstyle #2}$}
              \restrictionaux{#1}{#2}}
              {\setbox1\hbox{${\scriptstyle #1}_{\scriptscriptstyle #2}$}
              \restrictionaux{#1}{#2}}
              {\setbox1\hbox{${\scriptscriptstyle #1}_{\scriptscriptstyle #2}$}
              \restrictionaux{#1}{#2}}}
\def\restrictionaux#1#2{{#1\,\smash{\vrule height .8\ht1 depth .85\dp1}}_{\,#2}} 

\def\Res(#1,#2){ \restriction{#1}{#2} }

\def\mapping(#1,#2){:\fd(#1)\rightarrow \fd(#2)}
\def\rmc(#1,#2){RM(#1,#2)}
\def\form(#1,#2){H(#1,#2)}
\def\val(#1,#2){V(#1,#2)}
\def\boole(#1){ B(#1) }

\def\agl#1#2{\textsc{agl}(#1,#2)}

\def\stab(#1){\textsc{stab}(#1)}
\def\stablevel(#1,#2){\textsc{stab}_{#1}(#2)}

\def\rmq(#1,#2,#3){\rmc(#1,#2)/\rmc(#3,#2)}

\def\level(#1){\underset{#1}{=}}
\def\modlevel(#1){\underset{#1}{\equiv}}

\def\myref(#1){[\ref{#1}]}
\def\val(#1){{\rm val}(#1)}

\def\agl#1{{\mathfrak #1}}

\def\anf(#1){{\rm anf}(#1)}
\def\fix(#1,#2,#3,#4){{{\rm fix}^{#1,#2}_{#3}}{(#4)}}
\def\classe(#1,#2,#3){{{\mathcal T}(#1,#2,#3)}}

\def\stab(#1){\textsc{stab}(#1)}
\def\stablevel(#1,#2){\textsc{stab}^{#1}(#2)}
\def\orbitlevel(#1,#2){\textsc{orbit}^{#1}(#2)}
\def\stableveldim(#1,#2,#3){\textsc{stab}^{#1}_{#2}(#3)}
\def\stableveldeg(#1,#2,#3,#4){\textsc{stab}^{#1,#2}_{#3}(#4)}
\def\level(#1){\underset{#1}{\sim}}
\def\bound(#1,#2){\underset{#1}{\overset{#2}{\sim}}}
\def\modulo(#1,#2){\mod\rmc(#1,#2)}

\def\pow#1.#2{\tiny$10^{#1.#2}$}

\def\stab(#1){\textsc{stab}(#1)}
\def\stablevel(#1,#2){\textsc{stab}_{#1}(#2)}
\def\stab(#1){\textsc{stab}(#1)}
\def\stablevel(#1,#2){\textsc{stab}^{#1}(#2)}
\def\stableveldim(#1,#2,#3){\textsc{stab}^{#1}_{#2}(#3)}
\def\stableveldeg(#1,#2,#3,#4){\textsc{stab}^{#1,#2}_{#3}(#4)}

\title{Observation of known APNs}
\author{Valérie Gillot}
\author{Philippe Langevin}
\address{Imath, university of Toulon}
\email{\{valerie.gillot,philippe.langevin\}@univ-tln.fr}

\date{\today}

\begin{document}

\maketitle

\begin{abstract}
We present new invariants, \APN-extendibility criterion
and a backtracking approach to identify several numerical facts supporting the conjecture 
that the set of 6-bit \APN functions is limited to 14 \CCZ-classes.
\end{abstract}
    
\section{Introduction}

At the Finite Fields conference FQ9 in Dublin, K.~A.~Browning, J.~F.~Dillon, 
M.~T.~McQuistan, and A.~J.~Wolfe announced the discovery of an \APN permutation 
in dimension six \cite{DILLON}. We refer to this permutation as the 
\emph{Dublin permutation}, or equivalently to its \CCZ-class as the 
\emph{Dublin \CCZ-class}. Since this announcement, numerous attempts have been 
made to construct new \APN permutations in even dimension; however, no further 
examples have been found.

Subsequent work has significantly clarified the landscape of \APN functions 
in dimension six. In particular, the results of \cite{EDEL} provide a list of 
14 distinct \CCZ-classes of 6-bit \APN functions with representatives of algebraic 
degree at most three. Extensive numerical searches in dimension six, the 
classification of cubic \APN functions \cite{LANGEVIN}, dedicated searches for 
\APN permutations \cite{LEANDER}, and more recent investigations on switching 
\cite{LEANDER} strongly suggest that no additional \CCZ-classes of \APN functions 
exist in this dimension.

\begin{conjecture}
\label{ALLKNOWN}
All 6-bit \APN mappings are known.
\end{conjecture}

Proving Conjecture~\myref(ALLKNOWN) remains extremely challenging due to the 
astronomical size of the set of mappings from $\fd(6)$ to $\fd(6)$ and the lack 
of general theoretical tools for the study of \APN mappings. Even for Boolean 
functions, there is currently no effective algorithm to determine whether a 
given function can occur as a coordinate of an \APN mapping.

Any potential new example in dimension six would therefore have algebraic 
degree at least four. In this paper, we present numerical results supporting 
Conjecture~\myref(ALLKNOWN). Our approach relies on new invariants for the 
classification of Boolean and vectorial functions, which allow an in-depth 
exploration of \APN extensions of $(6,3)$-bent functions and $(6,4)$-mappings, 
in connection with the notion of switching.

We focus in particular on the possible existence of degree-four \APN functions 
with a specific structural property suggested by the decomposition of the 
14 known \CCZ-classes \cite{CALDERINI}. More precisely, 12 of these 14 classes 
contain at least one \CCZ-class of vectorial functions for which the set of 
fourth-order spectral moments of the components takes exactly two distinct 
values. We present a procedure to classify 6-bit \APN quartic functions sharing 
this regularity. To this end, we introduce a new algorithm to test the existence 
of an \APN extension of a given $(m,m-2)$-function.

\section{Boolean function}
Let $\field(2)$ be the finite field of order $2$. Let $m$ be a positive integer. 
We denote $\boole(m)$ the set of Boolean functions $f \colon \fd(m) \rightarrow\field(2)$. Every Boolean 
function has a unique algebraic reduced representation:
\begin{equation}\label{ANF}
f(x_1, x_2, \ldots, x_m ) = f(x) = \sum_{S\subseteq \{1,2,\ldots, m\}} a_S X_S,
\quad a_S\in\field(2), \ { X_S = \prod_{s\in S} x_s}.
\end{equation}

In this note for $m=6$, we use the letter $a,b,c,d,e,f$ to denote
the variables $x_1$, \ldots $x_6$. The \textsl{degree} of $f$ 
is the maximal cardinality of $S$ with  $a_S=1$ in the algebraic form, and  the minimal cardinality 
of $S$ such that $a_S=1$ is the \textsl{valuation} of $f$. Given integers, $0 \leq s\leq t\leq m$, we use
the  notation $\bst(s,t,m)$ to denote the space of Boolean functions of 
$m$ variables with valuation greater or equal to $s$, and degree less or
equal to $t$. The \textsl{Walsh coefficient} of $f\in \boole(m)$ at $a \in \fd(m)$ is 
$\tfr (f,a) =\sum_{x \in \fd(m)} (-1)^{f(x)+a.x}$, with $.$ the usual scalar product.  We denote
by $\kappa(f)$ the normalized the 4th-order spectral moment of $f$ :
\begin{equation}
\label{KAPPA}
 \kappa(f)= \frac 1q \moment(4, f)=\frac{1}{q^3}\sum_{a\in \fd(m)} {\tfr(f,a)}^4.
\end{equation}

The Walsh 
coefficients satisfy Parseval's identity:
\begin{equation}
\label{PARSEVAL}
 \sum_{a\in \fd(m)} {\tfr(f,a)}^2 = q^2,\quad q= 2^m.
\end{equation}

And giving a minoration  for the \textsl{linearity}  of $f$ :
\(
   \lin(f) := \max_{a\in\fd(m)} \abs({\tfr(f,a)})\geq 2^{m/2}\), 
a Boolean function achieves that bound is called \textsl{bent function}. These functions exist if and only if  $m$ is even, and  satisfy 
$\kappa(f)=1$. The \textsl{auto-correlation} 
of a Boolean function $f$ is defined for $t\in \fd(m)$ by:
\begin{equation}\label{CORR}
f\times f (t)=\sum_{x+y=t} (-1)^{f(x)+f(y)}= \frac 1 q \sum_{a\in \fd(m)} \tfr(f,a)^2 (-1)^{a.t}.
\end{equation}

\def\invnum(#1){\mathfrak n(#1)}
\def\invder(#1){\mathfrak D(#1)}
\def\invres(#1){\mathfrak R(#1)}

One says that the Boolean function $g$ is equivalent to
$f$ if there exists an element  $\agl s\in\aglgroup(m)$  
such that  $g = f\circ \agl s$.  The orbit of $f$ is precisely
the set of elements $g$ that are equivalent to $f$, it is
the orbit of $f$ under the action of $\aglgroup(m)$. Similarly,
we say that $g$ is ea-equivalent to $f$ if there exists an
element $\agl s \in\aglgroup( m )$ such that
$\deg(  g + f\circ A ) \leq 1 $.  For the dimension 6, we know 
\cite{VGPL-project, MAIORANA} there are $150{,}356$  ea-classes 
of Boolean function and a variant of the lift by restriction
of  \cite{BL2003} provides a very effective invariant  
that makes only one collision between the ea-class :

$$ab+bc+cd+bce+abde+abcf+adf\quad \text{and}\quad ac+bc+bd+bce+abde+abcf+adf.$$

The lift by restriction $J'$  of an invariant $J$ on $\boole(m-1)$
is an invariant on $\boole(m)$ defined as follow. For each $u \in\fd(m)$, 
we denote by 
\[
H_u := \{\, x \in \mathbb{F}_2^m \mid u.x= 0 \,\}, 
\qquad 
\overline{H}_u := \{\, x \in \mathbb{F}_2^m \mid  u.x = 1 \,\}
\]
the affine hyperplanes orthogonal to $u$ and its complementary. The 
\emph{restrictions} of $f$ to these hyperplanes says
\(f|_{H_u}\)   and  \(f|_{\overline{H}_u}\)
belong to $\boole(m-1)$. For each $u \in \fd(m)$, we consider 
the multiset 
$R_f( u ) := \{\!\!\{  J(  {\Res(f, H_u )} ), J( {\Res( f, \bar H_u ) } \}\!\!\}$,
and the distribution 
$J'(f) := \{\!\!\{  R_f(u)) | u\in\fd(m)  \}\!\!\}$
is an invariant on $\boole(m)$. In our implementation, 
a multiset $R_f(u)$ correspond dynamically to a number,  
in order to have a numerical function,
and the distribution of the Fourier coefficients, is
what we call a Fourier lift by derivation of $J$. For $m=5$, the number of extended–affine equivalence classes is 
$206$.  By numbering these classes via $\invw(f)$, 
one obtains a Fourier lift by restriction of $\invw$ takes 
$150\,355$ distinct values on $\bst(2,6,6)$. Our
implantation computes the invariant of all member 
in 18 seconds.

\section{vectorial function}

In that note, a vectorial $(m,n)$-function is a mapping from $\fd(m)$ into $\fd(n)$, 
defined by $n$ linearly independent \textsl{coordinate Boolean functions} 
such that $F(x) = \big(f_1(x), f_2(x), \ldots, f_n(x) \big)$. For any $0\ne b\in\fd(n)$,
the Boolean function $x\mapsto F_b(x) = b . F(x)$ is a \textsl{component} of $F$,
and we denote by $\comp(F)$ of the set of non zero components. Two function sharing
the same space of components modulo affine function are said congruent. We says that
that $F'\mapping(m,n)$ is  $k$-connected to $F$  by $G\mapping(m,n-k)$ if 
$$
    \comp(G) = \comp(F')\cap\comp(F), \quad \dim \comp(G)=n-k.
$$

When we observe the Boolean components space $\comp(F)$ of a vectorial function $F$, we are interested on the one hand in the set of their classes $C_F:=\{ \invw( f ) \mid  f \in \comp(F)\}$ and on the other hand in the set of all normalised 4th-order spectral moments $K_F:=\{ \kappa(f) \mid 0 \ne f \in \comp(F)\}$. For these two sets, we also study their cardinality and their distribution of values. Note that $\sharp K_F\leq \sharp C_F$.
We are particularly interested in vectorial functions such that $\kappa(f)$ take one or two values. 
We say that $F$ is homogeneous if $C_F =1$ and bihomogeneous if  $C_F = 2$.
A vectorial function $F$ has $k$ levels of 4-th order spectral moments if the cardinality of $K_F$ is $k$.
  
\begin{example}
If $m$ is odd then all the non-zero components of the power function $x^3$ in $\ffd(m)$ 
are equivalent, it is an homogeneous map, $\sharp C_F=1$, and thus $\sharp K_F=1$. 
\end{example}

\begin{table}
\begin{tabular}{|ccr|ccr|ccr|}
\hline
    degree  &$\kappa$ &class & degree  &$\kappa$ &class & degree  &$\kappa$ &class \\
\hline
2 & 1.0000 &       1 & 5 & 1.6562 &      13 & 6 & 1.6211 &      11 \\
3 & 1.0000 &       3 & 5 & 1.7500 &      17 & 6 & 1.7148 &      10 \\
4 & 1.7500 &       8 & 5 & 1.8438 &      84 & 6 & 1.8086 &      56 \\
4 & 1.9375 &      54 & 5 & 1.9375 &     213 & 6 & 1.9023 &     161 \\
5 & 1.4688 &       8 & 6 & 1.2461 &       4 & 6 & 1.9961 &     444 \\
\hline
\end{tabular}
\caption{\label{NORMALIZED} 15 normalized moment of order 4 smaller that 2.}
\end{table}

We say that $F\mapping(m,n)$ is an extension of $G\mapping(m,n-k)$ 
by $g\mapping(m,k)$ 
if the space $\comp(G)$ is a subspace of codimension $k$ of $\comp(F)$, 
such that $\comp(F) = \comp(G)\oplus \comp(g)$. We say $F$ and $G$ are affine equivalent if there exist an affine $(m,m)$-permutation $A$, 
an affine $(n,n)$-permutation $B$ such that $\deg( G  + B\circ F \circ A ) \leq 1$ ;
For the numerical exploration of 6-bit vectorial
function, we construct a invariant as follows. For a vectorial $(5,6)$-function $F$, 
we define the invariant by composition 
of $F$ as the multiset :

$$\invj( F ) = \{\!\!\{ \invw( f  ) \mid f \in \comp(F) \}\!\!\}$$

is an \EA-invariant. The lift by restriction, say $\invj'$ of this invariant
is a nice invariant that takes 714 distinct values on the set
of known 6-bit APN. 
The collisions concern three functions of degree 2. The Walsh coefficient of a $(m,n)$-Function $F$ at $(a,b) \in \fd(m)\times \fd(n)$ 
is Walsh coefficient of its component $F_b$:
$$\TFR (F,a, b) =\tfr (F_b,a)=\sum_{x \in \fd(m)} (-1)^{F_b(x)+a.x}.$$ 

The degree and linearity are extended 
to  an $(m,n)$-function $F$ :
\[
\deg(F):=\max_{  f \in\comp(F)} \deg(f), \quad \lin(F):=\max_{ 0\ne f \in\comp(F)} \lin(f)
\]
A $(m,n)$-function is a \textsl{bent} if all its non-zero components are bent. A such function exists if and only if $m$ is even and $n\leq m/2$. For $m=2k$ and $n>k$, an $(m,n)$-function 
$F$ is called $(m,n)$-\MNBC function see \cite{SACHAMNBC}, 
if it has the maximum number of bent components $2^n-2^{n-k}$.
We know there are precisely 13 EA-classes of (6,3)-bent function 
in dimension 6.

\section{Almost perfect function}

Let \rouge{$F$ be an} $(m,n)$-function. For $0\ne u\in\fd(m)$, $v\in \fd(n)$, we denote $N_F(u,v)$ the number of solutions in $\fd(m)$ of the equation $F(x+u)+F(x)=v$. Note that if $x$ is a solution then $x+u$ is also a solution. Thus, $N_F(u,v)$ is even. 
\begin{equation}
 N_F(u,v)=\frac 1 {2^{m+n}} \sum_{a\in \fd(m),b\in\fd(n)} \tfr (F_b,a)^2 (-1)^{a.u} (-1)^{b.v} 
 = \frac 1 {2^n} \sum_{b\in\fd(n)} F_b\times F_b (u) (-1)^{b.v}.
\end{equation}
The \textsl{differential uniformity} of a $(m,n)$-function $F$ is  $\Delta_F:= \max\limits_{u\in\fd(m)\setminus \{0\}, v\in \fd(n)} N_F(u,v).$
A $(m,m)$-function $F$ is \textsl{almost perfect non linear} (\APN) 
if and only if
it satisfies one of the following equivalent properties:
\begin{enumerate}[(i)]
\item The differential uniformity of $F$ is $\Delta_F=2$;\\
\item For all 2-flat $\{x,y,z,t\}\subseteq \fd(m)$,  $F(x) + F(y) + F(z) + F(t) \ne 0$;\\
\item $\sum\limits_{0\ne f\in \comp(F)} \kappa(f) = 2(q-1)$,
where $q=2^m$.
\end{enumerate}

Point~(iii) shows that an \APN\ function 
has at least one component whose normalized 
spectral moment is at most 2. 
In dimension~6, this property is satisfied 
by 66 ea-classes of degree less or equal
to 4 (see Table~\ref{NORMALIZED}).

\begin{lemma} If $F$ is \APN in even dimension then \textcolor{black}{$\sharp K_F\geq 2$}.
\end{lemma}
\begin{proof}
If $f$ is a non-zero component of $F$ with $\sharp K_F = 1$, then (iv) implies that $\kappa(f)= 2$. \textcolor{black}{By little Fermat's Theorem $q^3\equiv q \mod 3$ and $\tfr(f,a)^4\equiv \tfr(f,a)^2 \mod 3$, applying Parseval's identity, we obtain $\kappa(f)\equiv q \mod 3$, that implies $m$ odd.}
\end{proof}

\textcolor{black}{The non-existence in even dimension of \rouge{\APN functions} with a single spectral moment of order 4 naturally leads us in the next section to look for \APN functions with two spectral moments of order 4. 
We introduce here the terminology of function with 2-spectral levels to designate a vectorial function $F$ such that $\sharp K_F=2$.}

For an \APN vectorial function $F$, we introduce  Boolean functions called \textsl{counting functions}. 
For a given $u\in \fd(m)$ defined for $v\in\fd(m)$
by    
\[
n^F_{u}(v)=\begin{cases}
1, & \text{if } N_F(u,v)=2;\\
0, & \text{if } N_F(u,v)=0.\\
\end{cases}
\]

For all $u\in\fd(m)\setminus\{0\}$, the counting function $n^F_u$ is balanced 
and the walsh coefficients are :

\begin{equation}\label{LINK}
    \forall b\in\fd(m)\setminus\{0\},\quad \tfr (n^F_{u},b)= -F_b\times F_b (u)
    \quad \text{and}\quad \tfr (n^F_{u},0)=0.
\end{equation}

The graph of $F\mapping(m,n)$ is the set 
$\Gamma(F)=\{ (x,F(x) ) \mid x\in\fd(m) \} \subset \fd(m)\times \fd(n)$. One says that two maps $F$, $G$ are \textsl{\CCZ-equivalent} if the
Boolean indicator of the graph are affine equivalent in $\boole(m+n)$,
there exists an affine transformation in $A\in\aglgroup(m+n)$ 
mapping $\Gamma(F)$ onto $\Gamma(G)$. 

The degree of known 6-bit \APN-mappings   takes values 2, 3 or 4. Their
set splits  in 716 \EA-classes which are grouped into 14 \CCZ-classes, 
see \cite{CALDERINI}. In total, 534 are mapping of degree 3. All these
14 classes, have maximal degree 4 and all have minimal degree 2,  except  
the cubic of Edel-Pott \cite{EDEL}. 

\begin{table}
\small
\begin{tabular}{r|T{11cm}|ccc|c}
\centering
\textbf{label} & \textbf{Code 64} &\multicolumn{3}{|c|}{degree}&$\Gamma$-rank\\ \hline
1  & \seqsplit{AAA0AEK2AUIoMcOmASMvUCS5FDB2dfTsAYqG7zfvWa0Mh9NpEOncu0DkXJ8Tx/Un} &1 &66 &18 &1170\\
2  & \seqsplit{AAAkA4YEAgw0WO+CAICuE0eK9VPDv/FxASQmrBjtYq4ulvdzUOG47Zx3xLTNIKyU} &1 &66 &18 &1174\\
3  & \seqsplit{AAAEAYIUAQg0We+yA4yOU0uKtF/Tvf1BACQWbBDdIa4uFP9zE+mYLphHhLjN4KyE} &1 &66 &24 &1170\\
4  & \seqsplit{AAABAoQJAhI82PuyAXDzMDv3z3sb5VGpZVhspNBU1YF9zGTDLQwM30s4UczIuOpK} &1  &9 & 3 & 1166\\
5  & \seqsplit{AAAkAw8oAow8qymaAQSmU06+1NXLLDV5ACQ2zBfJEukqdHB/cOeo7ZFDtXfBgquA} &1 &69 &16 &1172\\
6  & \seqsplit{AAAkAYwMA4o0iC6+AwaOEsuiVdnLzjxFAK4WTBbtcuMatHNDQqysHlVTZbT1s2Wo} &0 &10 &15 &1300\\
7  & \seqsplit{AAAEAIw8Ao4U2W+aAYOSo4WC1FD3rTtRASg2bBLVUuMy5LRncWy8vtx39frN4Sew} &1 &15  &3 &1158\\
8  & \seqsplit{AAAkAQsYAYQsSauCAgim0E6uldXLDrdRACgGzh/JUOka1/pHc+eYbp1jtX/h4SGI} &1 &63 &27 & 1170\\
9  & \seqsplit{AAAkAosgA4g8CSO6AYK28MaOtNHDTbV5AKwejB/5UmES1vJ3EW+IbhNT9XnpgiWw} &1 &15  &3 & 1170\\
10 & \seqsplit{AAAkAw8oAIoEC6WKAgaeEUiWlNXbj7tRA64mLBPh0Gky9/R3UO2IbxFLFXP5Iq+4} &1 & 1  &1&1146\\
11 & \seqsplit{AAAEAwgUAoQ8OW+iAICO0MWq9dvLHX1hAqoGbBTNgiYe1HtbU2+Y7pxnJDz9oSyM} &1  &5  &7& 1166\\
12 & \seqsplit{AAAkAYQsA4o0SyquAIyes8O6td3jT7ZVACY+bBTtUuk6d/97Qa6Un1drpbr9MmeQ} &1 &66 &24 &1168\\
13 & \seqsplit{AAAgAIQ4AEY8COKmAwaKM0GelRnzrX5lAyUG7B/lI+ESxPtzgiuMXdJjNLb942+Q} &1 & 1  &1 &1102\\
14 & \seqsplit{AAAgAoYQAEQ0qGiuAwaKEcG+lRv7LXZlA6cGTBXlgesyZPN7ICOkf9BDNDb1wW+4} &1 &69 &22 &1172\\
\end{tabular}
\caption{\label{K64} Code-64 of representatives of 14 known \CCZ-class of 6-bit 
\APN identified by degree distribution of EA-class and $\Gamma$-rank}
\end{table}
Using the degree distribution of EA-class, and $\Gamma$-rank
one can identify the 14 known classes. However, \CCZ-invariant 
are not legion and it is a intresting question 
to find a discriminant set of numerical invariants.
In  \cite{NIKOLAY},  it is pointed out that
 the combination of  $\Gamma$-rank and   $\Delta$-rank. 
takes 10 values. In next lines, we define 3 news useful
affine invariant on Boolean function that enable us to discriminate  
the graph of \APN-mappings.

\begin{table}[h]
\begin{center}
\begin{tabular}{|l|c|c|c|c|c|}
\hline
     invariant &$\invmul(2,0)$   &  $\invval$   &$\Gamma$-rank  &$\Delta$-rank  &$\invtrans(9)$      \\
     \hline
    number of value  &6                  &5           &9                 &3    &3\\
    time in sec.  &1.24               & 0.21        &56               &42   &2080\\
     \hline
\end{tabular}
\caption{Cumulative computational cost of invariants 
of the 14 known class.}
\end{center}
\end{table}

\begin{invariant}[multiplicative]
For a Boolean function $f\mapping(m,1)$, we denote by
$\endomul(p,q,f)$ the multiplication endomorphism that
maps $g\in\rmc(0,p)$ to the product $gf$ projected onto
$\bst(p,m,m)$. The dimension of the kernel of $\endomul(p,q,f)$
is a numerical invariant, called the multiplicative invariant
and denoted $\invmul(p,q)$.
For example, $\invmul(2,0)$ takes 6 values
on the graphs of the 14 known classes.
\end{invariant}

\begin{invariant}[valuation]
We construct an invariant dynamically by examining the distribution
of the Fourier coefficients of a Boolean function $f$. For each
integer $v \in \{ \infty, 0, 1, \ldots, m \}$, let $\psi_v$ denote the indicator function of the Walsh coefficients
with valuation $v$. Given an invariant $\invj$, we define 
the valuation invariant as the tuple:

$$
    \invval(f) = \big( \invj(\psi_\infty), \invj(\psi_0), \ldots, \invj(\psi_m) \big).
$$
\end{invariant}

\begin{invariant}[transvection-like]
The set of quadratic functions of the form 
$T_{\phi,u}\colon x \mapsto x + \phi(x)\cdot u$, where $q\in\rmc(2,m)$
and $u\in\fd(m)$, is stable under the action of the group $\aglgroup(m)$. 
The number of elements of this form in $\stablevel(r,f)$
defines an invariant, denoted $\invtrans(r)$.
\end{invariant}

\begin{problem}
    The combination of these invariants is discriminative on the 14 known classes, and without  $\invtrans(2)$, we have 1 collision.
    The challenge is to find a faster discriminative solution to bypass the bottleneck 
    when going from 13 to 14 classes.
\end{problem}

\section{Observation of known APN}

The known 6-bit APN functions have algebraic degrees 2, 3, or 4, and they are distributed
into 716 EA-classes and 14 CCZ-classes \cite{CALDERINI}. There are 534 classes of
cubic functions, and each CCZ-class has at least one 
cubic representative and one quartic representative. Every CCZ-class contains a quadratic representative,
except for the Edel-Pott class \cite{EDEL}. 

\begin{question}
    Does there exist an APN function outside the 14 known classes?
\end{question}

\subsection{component of \APN-mapping}

The study for quadratic functions was carried out by \cite{LANGEVIN}, while that for quartic functions still seems out of reach. 
Since the degree is not a CCZ-invariant, the question concerns the existence of an APN function $F$ such that

$$
        4 \leq \degmin(F) := \min_{\cczeq(F,G)} \deg(G) ?
$$

It is quite remarkable that only 37 ea-classes
of Boolean functions appear in the components
of the known \APN functions, 
see  Table \ref{37}, corresponding
to 8 normalized moments.

\begin{fact}
The set of components of APN functions is reduced to 37 elements.
\end{fact}

\def\foo(#1,#2){${#1}_{#2}$}
\begin{table}[h]
\centering
\scriptsize
\begin{tabular}{|c|c|c|c|l|}
\hline
$\kappa$  &$\sharp$   &   &  &\textbf{distribution}\\
\hline
1.00 &4 &4 &0  &\foo(596,1 ) \foo(8230,3 )    \foo(6951,6 ) \foo(1546,9 )\\
1.75 &4 &8 &17 &\foo(2112,12 ) \foo(2214,13 ) \foo(476,14 ) \foo(176,22 ) \foo(36,23 )  \foo(34,27 )\\
2.50 &15 &216  &4454&\foo(968,10 ) \foo(1224,11 ) \foo(12,20 ) \foo(4,21 )\foo(72,24 ) \foo(144,25 ) \foo(20,26 )
\foo(72,28 ) \foo(8,29 ) \foo(8,30 ) \foo(4,32 ) \foo(4,33 ) \foo(6830,4 ) \foo(7014,5 )  \foo(1380,8 )\\
3.25 & 4&191&3261&\foo(152,15 )  \foo(216,16 )  \foo(168,17 )  \foo(8,31 )\\
4.00 & 4&86&1056&\foo(780,0 )  \foo(72,18 )  \foo(52,19 ) \foo(1689,2 ) \foo(1796,7 )\\
4.75 &1&17&258& \foo(12,36 )\\
8.50 &1&3&8 &\foo(24,35 )\\
16.00 & 1&1&0&\foo(4,34 )\\
\hline
\end{tabular}
\caption{\label{37} Repartition of components of known APN.}
\end{table}

\begin{remark}
    Detecting a new \APN function in a massive dataset can be facilitated  
    by using the knowledge of the 37 known components.
\end{remark}

\subsection{counting function}

The counting functions of the known \APN functions 
share 13 affine classes of Boolean functions.  It is observed 
that all of them have degree $\leq 4$,  that there is always 
at least one of degree 1 and none are cubic.

\begin{table}[h]
\centering
\begin{tabular}{|r|r|c|c|c|}
\hline
\textbf{mult.} & \textbf{$\deg$} & \textbf{1} & \textbf{2} &\textbf{4} \\
\hline
13 & 2 & 63  &&\\
3  & 3 & 13  &50  &\\
3  & 3 & 15  &48  &\\
1  & 3 & 21  & 42  &\\
114 & 3 & 3  &60  &\\
219 & 3 & 5  &58  &\\
117 & 3 & 7  &56  &\\
31  & 3 & 9  &54  &\\
2   & 4 & 7  & &56  \\
\hline
\end{tabular}
\caption{Summary of degrees of APN and counting functions.}
\end{table}

\begin{fact}
Known 6 bit \APN functions contain at least one affine counting function.
\end{fact}

We apply the relation \ref{LINK}, one can 
see that all the counting functions of 
a putative  6-bit \APN of degree 6 
are quintic.  A fact that provides  supporting evidence 
for the widely accepted~:
\begin{conjecture}
    The degree of an \APN mapping is less than $m$.
\end{conjecture}

\subsection{low level mapping}

\begin{proposition}
All 6-bit APN functions of type $(1,4)$ are quadratic.
\end{proposition}

\begin{proof}
The component space of a 6-bit APN function $F$ of type $(1,4)$
contains 56 bent functions forming a system of rank~6.
Since bent functions have algebraic degree at most~3,
we obtain $\deg(F) \leq 3$, and all \APN cubic functions are 
already classified in \cite{LANGEVIN}.
\end{proof}

The \CCZ-class that contains the well-known Dublin permutation 
splits into 13 distinct classes. A finer analysis of the component 
functions is given in the table below.

\begin{wraptable}{l}{70mm}
\begin{tabular}{|c|ccc|cccc|}
\hline
\multirow{2}{*}{\#} & \multicolumn{3}{c|}{degree} & \multicolumn{4}{c|}{ 4th-spectral moment} \\
\cline{2-8}
&2&3&4&1&1.75&2.5&4.0\\
\hline
1  &63 &   &    & 42[1] &     &    & 21[1]\\
2  &    & 7 & 56 &    & 56[1]   &    & 7[\textcolor{green}{1}]\\
5  &1  & 62 &    & 30[2] &  & 24[2] & 9[3]\\
2  &    & 31 & 32 & 12[2] & 32[3] & 12[2] &7[\textcolor{green}{2}]\\
1  &    & 31 & 32 & 12[1] & 32[2] & 12[2] &7[\textcolor{green}{2}]\\
2  &    & 31 & 32 & 12[2] & 32[2] & 12[2] &7[\textcolor{green}{2}]\\
\hline
\end{tabular}
\end{wraptable}
\small

There are 2-spectral levels APN functions. The \CCZ-class 
of Dublin permutation is divided into 13 \EA-classes \cite{CALDERINI}, 
3 of which have 2 spectral levels, see line 1 and 2. 
This table also gives the distribution of the degrees and spectral levels of the components, specifying the number of \EA-classes. For example, the line 2 there is 2 \EA-classes which one that contains 
the Dublin permutation, 7 components are cubics and 56 are quartics. 

\small
\noindent Moreover,  56 components with  spectral moment 1.75 are in [1] \EA-class,
7 components with  spectral moment 4.0 are in [1]. Does there exist a two-level 
\APN function outside the 14 known classes? We decided to search
for other examples by extension process. A vectorial \APN function $F$ 
is with \textsl{2-spectral levels} if the normalized 4th-order spectral 
moments of its components take 2 distinct values $\alpha$ and $\beta$. 
In this case, \begin{equation}\label{AB}
        \alpha A + \beta B = 2(q-1), \quad A + B = q-1 ;
\end{equation} where $A$ (resp. $B$) is the number of components $f$ of $F$ such that $\kappa(f)=\alpha$ (resp. $\kappa(f)=\beta$). We suppose that $\alpha < \beta$ and we say $F$ is a function of type $(\alpha,\beta)$.
Using the classification of Boolean functions,  among 293 values of $\kappa$, we found 62 possible pairs satisfying \ref{AB}~, involving function of degree less or equal to 4:
$$
\hbox{
\tiny
\begin{tabular}{|r|c|c|c|c|r|c|c|c|c|}
\hline
$\alpha$  &A&$\deg$&$\sharp$& $\beta$ &B&$\deg$&$\sharp$\\
\hline
\rowcolor[gray]{.8} 
\rowcolor{green!20}1.0&42 &23. &{4} &4.0&21 &234 &{86}\\
\rowcolor{bad!20} 1.0&60 &23. &{4} &22.0& 3 &.3. &{1}\\
\rowcolor{bad!20} 1.0&56 &23. &{4} &10.0& 7 &.3. &{1}\\
\rowcolor{bad!20} 1.0&49 &23. &{4} &5.50&14 &.34 &{29}\\
\rowcolor{bad!20} 1.0&21 &23. &{4} &2.50&42 &.34 &{216}\\
\rowcolor{bad!20} 1.0&57 &23. &{4} &11.50& 6 &.34 &{5}\\
\rowcolor{bad!20} 1.0&35 &23. &{4} &3.250&28 &.34 &{191}\\
\rowcolor[gray]{.8}1.0&15 &23. &{4} &2.3125 &48 &.34 &{214}\\
\rowcolor{bad!20} 1.0&51 &23. &{4} &6.250&12 &..4 &{13}\\
\rowcolor{bad!20} 1.0&47 &23. &{4} &4.9375 &16 &..4 &{37}\\
\rowcolor{bad!20} 1.0&39 &23. &{4} &3.6250&24 &..4 &{67}\\
\rowcolor[gray]{.8}1.0& 7 &23. &{4} &2.1250&56 &..4 &{49}\\
\rowcolor{bad!20} 1.0&55 &23. &{4} &8.8750& 8 &..4 &{2}\\
\hline
\end{tabular}
\begin{tabular}{|r|c|c|c|c|r|c|c|c|c|}
\hline
$\alpha$  &A&$\deg$&$\sharp$& $\beta$ &B&$\deg$&$\sharp$\\
\hline
\rowcolor{green!20} 1.750&56 &..4 &{8} &4.0& 7 &234 &{86}\\
1.750&42 &..4 &{8} &2.50&21 &.34 &{216}\\
1.750&60 &..4 &{8} &7.0& 3 &.34 &{3}\\
1.750&51 &..4 &{8} &3.0625 &12 &.34 &{321}\\
1.750&35 &..4 &{8} &2.3125 &28 &.34 &{214}\\
1.750&59 &..4 &{8} &5.6875 & 4 &..4 &{25}\\
1.750&21 &..4 &{8} &2.1250&42 &..4 &{49}\\
1.750&49 &..4 &{8} &2.8750&14 &..4 &{119}\\
1.750&57 &..4 &{8} &4.3750& 6 &..4 &{34}\\
\rowcolor[gray]{.8}1.9375 &56 &..4 &{54} &2.50& 7 &.34 &{216}\\
1.9375 &60 &..4 &{54} &3.250& 3 &.34 &{191}\\
1.9375 &42 &..4 &{54} &2.1250&21 &..4 &{49}\\
\rowcolor[gray]{.8}1.9375 &62 &..4 &{54} &5.8750& 1 &..4 &{19}\\
\hline
\end{tabular}
}
$$

\begin{remark}
Observe that a function of type $(1,\beta)$ with more than 32 bent components cannot exist.
Indeed, the rank of the bent components would necessarily be 6, which would imply that the function is cubic.
However, \APN cubic functions are already completely classified.
\end{remark}

\section{Switching known APNs}

Let us recall a flat  is a 4-set $\{x,y,z,t\}\subset \fd(m) $
such that that $x+y+z+t=0$, and $F\mapping(m,m)$ is \APN 
if and only if $F(x) + F(y)+F(z)+F(t) \ne 0$ for all 2-flat  $\{x,y,z,t\}$.
Given an integer $k\leq m$, and  $G\mapping(m,m-k)$, we define the set 
of 2-flats of $G$:
$$
                \flat(G)  := \big\{   \{x,y,z,t\}  \mid G(x) + G(y) + G(z) + G(t) = 0 \big\}
$$

One see that the  extension $F:=(G,g)$ of $G$  by $g\colon\fd(m)\rightarrow \fd(k)$
mapping $x\mapsto \big(G(x), g(x)\big)$ is \APN if and only if  

\begin{equation*}
        \forall \{x,y,z,t\}  \in\flat(G)  \qquad g(x) + g(y) + g(z) + g(t) \not= 0 .
\end{equation*}

In other word,  denoting by $T=\fd(n) \setminus \{0\}$,  
$g$ is in the preimage of $T^N$, by the linear endomorphism:

$$
            \Sha( g ) := ( \ldots, g(x)+g(y)+g(z)+g(t), \ldots )\in\field(2)^N, \quad N=\card( {\flat(G)} )
$$

The determination of \APN extensions of $G\mapping(m,m-1)$ is a easy task
that can be done by solving a linear system, and becomes harder when $k>1$. 
Given an \APN-function $F$,  one may construct new one by switching and more
generally finding all the maps that are 1-connected to $F$. We say that $F$
is 1-switchable by $G$ if there exist $F'\not\equiv F$ which is 1-connected 
to $F$.

\begin{lemma}
\label{KERNEL}
An \APN mapping $F\mapping(m,m)$ is  1-switchable 
by $G$ if and only if  $\dim K(G) > 2m$.
\end{lemma}
\begin{proof}
\end{proof}

We now, explain how one can use invariant to prove
that the set of  716  konwn  \APN-mapping is closed 
by 1-connection.  Foreach 
ea-class $F$ of konown APNs, we conctruct all the 63 spaces $S$ of codimenion
1 in $\comp(F)$, considering $G\mapping(6,5)$ with support $S$ 
and we use Lemma  \ref{KERNEL} to retains only the maps
that are not congruent to $F$.

\begin{minipage}{60mm}
\includegraphics[scale=0.2]{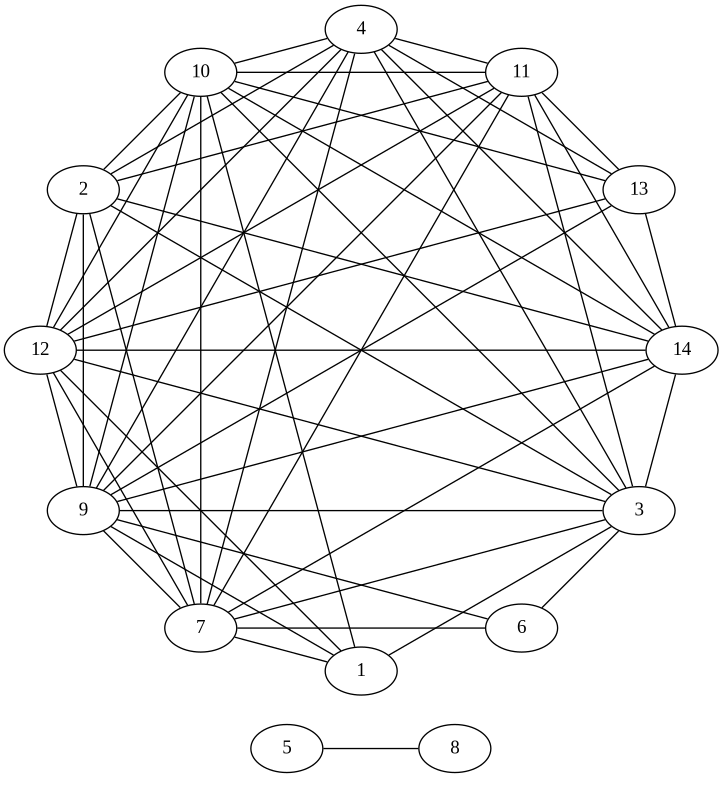}
\end{minipage}
\begin{minipage}{60mm}
\centering
\scriptsize
\setlength{\tabcolsep}{4pt} 
\renewcommand{\arraystretch}{0.9} 
\begin{tabular}{|c|c|c|c|c|ccccc|}
\hline
\multirow{2}{*}{\textbf{label}}&
\multirow{2}{*}{\textbf{ea}} &
\multirow{2}{*}{\textbf{sub}} &
\multirow{2}{*}{\textbf{apn}} &
\multirow{2}{*}{\textbf{class}} &
\multicolumn{5}{c|}{\textbf{dimensions} }\\ \cline{6-10}
 & & & & & \textbf{1} & \textbf{2} & \textbf{3} & \textbf{4} & \textbf{5} \\
\hline
1  & 25 & 1575 & 100 & 60  & 12 &   &   & 22 &   \\
2  & 19 & 1197 & 453 & 206 & 355 &   & 26 &   & 4 \\
3  & 91 & 5733 & 1175 & 457 & 644 & 34 & 69 & 44 & 16 \\
4  & 19 & 1197 & 474 & 243 & 35 & 41 & 79 &   & 24 \\
5  & 3  & 189  & 66  & 4   & 66 &   &   &   &   \\
6  & 13 & 819  & 162 & 67  & 162 &   &   &   &   \\
7  & 91 & 5733 & 1319 & 526 & 813 & 28 & 94 & 22 & 16 \\
8  & 3  & 189  & 66  & 5   & 66 &   &   &   &   \\
9  & 92 & 5796 & 1584 & 516 & 813 & 152 & 97 & 44 &   \\
10 & 85 & 5355 & 1413 & 507 & 690 & 121 & 75 & 44 & 16 \\
11 & 85 & 5355 & 1434 & 504 & 758 & 124 & 116 &   & 16 \\
12 & 91 & 5733 & 1427 & 492 & 791 & 93 & 94 & 22 & 16 \\
13 & 13 & 819  & 336 & 163 & 108 & 114 &   &   &   \\
14 & 86 & 5418 & 1589 & 458 & 773 & 183 & 150 &   &   \\
\hline
\end{tabular}
\end{minipage}

The results are summarized in the table above. For each \CCZ-class with a given label, we recall the number of \EA-classes, indicate the number of subfunctions to be extended, the number of APN functions obtained by extension, and finally the number of \EA-classes.
The last five columns record the multiplicities of the dimensions  $K(G)$ encountered during the process.

\begin{fact}
The set of known \APN, is closed under 1-switching. The 1-connection graph  
of \CCZ-class has order 14 and two connected components.
\end{fact}

\begin{table}
\centering
\scriptsize
\begin{tabular}{|c|c|c|c|c|c|c|c|c|c|c|c|c|c|c|}
\hline
\multicolumn{15}{|c|}{14 \CCZ-class of known \APN}\\
\hline
label  &1  &2&3&4 &5 &6         &7&8&9&10&11&12&13&14\\
\hline
ea & \textcolor{red}{25} 
    & 19 & 91 & 19 & 3 & \textcolor{blue}{13} & 91 & 3 & 92 & 85 & 85 & 91 & 13 & 86 \\
\hline
switch  &16275
        &12369
        &59241
        &12369 
        &1953 
        &8463        
        &59241 
        &1953
        &59892
        &55335
        &55335
        &59241
        &8463
        &55986\\
class  &5945  
       &2133 
       &8263
       &1699
       &193 
       &1111       
       &8493
       &65
       &8315
       &8188
       &7878
       &8170
       &1584
       &7341\\
\hline
apns    &41626 
        &85271 
        &154303
        &63682
        &6069
        &34076
        &154832
        &2928 
        &156167
        &154780
        &147546
        &161340
        &55048
        &144378\\
\hline
class   &519  
         &700
         &704
         &694 
         &473 
         &701        
         &704
         &353
         &706
         &703
         &704
         &704
         &688
         &706\\
\hline
\end{tabular}
\end{table}

\section{switching by backtraking}

To determine the \APN extensions of $G\mapping(m,m-2)$
by $g\mapping(m,2)$, we use a backtracking approach with 
a carefully chosen initialization. We consider a system of rank $2m-1$ 
composed of a basis of $\comp(F)$ and $m+1$ independent affine functions, 
equipped with a set of "pivots" $\{p_k\}_{k=1}^{2m-1}$ such that $\phi_i(p_j) = \delta{i,j}$.
This allows us to restrict the search to functions $g$ that vanish on the entire set of pivots, 
and if $p$ is not a pivot, we may further assume that $g(p) \in {00,01}$. 
For $m = 6$, one must instantiate $g$ on the remaining 53 vectors. 
The computational effort is substantial and requires an implementation 
based on bit programming, parallelized over a large number of cores.

A vectorial function defines $\binogauss(6,2) = 615$ 
subfunctions of codimension 2, which leads us to run $466{,}116$ 
backtracking procedures. To speed up the process, we applied the invariant $J$ 
to the set of $(6,4)$-subfunctions of a \CCZ-class before computing 
all \APN extensions.

\begin{table}[ht]
\centering
\captionsetup{font=small, labelfont=bf}
\caption{14 CCZ classes of known APN functions}
\scriptsize 
\setlength{\tabcolsep}{4pt} 
\renewcommand{\arraystretch}{0.95} 

\begin{tabular}{c|*{14}{c}} 
\toprule
\multicolumn{15}{c}{\textbf{14 CCZ classes of known APN}} \\
\midrule
\textbf{num} & 1 & 2 & 3 & 4 & 5 & 6 & 7 & 8 & 9 & 10 & 11 & 12 & 13 & 14 \\ 
\midrule
\textbf{ea} &
\textcolor{red}{25} & 19 & 91 & 19 & 3 & \textcolor{blue}{13} & 91 & 3 & 92 & 85 & 85 & 91 & 13 & 86 \\
\midrule
\textbf{sub} &
16275 & 12369 & 59241 & 12369 & 1953 & 8463 & 59241 & 1953 & 59892 & 55335 & 55335 & 59241 & 8463 & 55986 \\
\textbf{filter} &
5945 & 2133 & 8263 & 1699 & 193 & 1111 & 8493 & 65 & 8315 & 8188 & 7878 & 8170 & 1584 & 7341 \\
\midrule
\textbf{apns} &
41626 & 85271 & 154303 & 63682 & 6069 & 34076 & 154832 & 2928 & 156167 & 154780 & 147546 & 161340 & 55048 & 144378 \\
\midrule
\textbf{classe} &
519 & 700 & 704 & 694 & 473 & 701 & 704 & 353 & 706 & 703 & 704 & 704 & 688 & 706 \\
\bottomrule
\end{tabular}
\end{table}

For each function $F$ in one of the 14 \CCZ-classes, 
we generated all codimension-2 subspaces  of the component spaces 
of functions equivalent to $F$. We filtered these using the invariant $\invj$  
and determined all extensions by backtracking. This yields 37 classes at the component level. 
Over the set of $1{,}362{,}046$ \APN functions, 
our invariant takes 714 distinct values.

\begin{fact}
    The set of known \APN is closed by 2-switch. 
\end{fact}

\section{two level}

It is alreadly a very hard task to construct  all $(\alpha,\beta)$
\APN map for an admissible pair.  
We decided to restrict our search to two-level vectorial functions 
containing two classes, one cubic and one quaratic. In this case, 
the cubic class forms an isospace, and there are four possible 
configurations, which are described in Table~\ref{PAIR}. 
The two-level Dublin classes of type  $(1.75,4)$
fall into this category.

\begin{table}[htbp]
\small
\caption{\label{PAIR} 4 numerical space compatible.}
\begin{tabular}{|l|c|c|c|l|c|c|c|c|}
\hline
$\alpha$ &$A$ &degree  &$\sharp$classes &$\beta$  &$B$  &degree &$\sharp$classes\\ 
\hline
\hline
 1.0000 &( 7) &23... &{4} & 2.1250 &(56) &..4.. &{49}\\
\rowcolor[gray]{.8}  1.7500 &(56) &..45. &{25} & 4.0000 &( 7) &23... &{6}\\
 1.9375 &(56) &..4.. &{54} & 2.5000 &( 7) &.3... &{5}\\
 1.9375 &(62) &..4.. &{54} & 5.8750 &( 1) &..4.. &{19}\\
 \hline
\end{tabular}
\end{table}

Starting from a cubic function $g$ and a quartic function $h$ 
such that $56 \times \kappa(g) + 7 \times \kappa(h) = 126$, 
we construct the \APN extensions in several steps. In the first step, 
we build a list of compatible cubics $L$, from which we extract a system o
f representatives $R$ under the action of the stabilizer of $h$:

\[
    L := \{ f \in \orbit(g) \mid f + h \in \orbit(h) \}, \qquad R := L \backslash \stable(h).
\]

We then construct the list of all subspaces
\[
    V := \{ \langle f_1, f_2, f_3 \rangle \mid f_1 \in R, \quad f_2 \in L, \quad f_3 \in L \}.
\]

Next, we construct the list of candidate $(6,4)$-functions $(f_1, f_2, f_3, h)$ 
by taking $\langle f_1, f_2, f_3 \rangle$, retaining only those candidates that satisfy 
the extendibility criteria given by Lemma \ref{DELTA} and Lemma \ref{SYSTEM}. If the set of candidates is too large, it is filtered 
using the invariant $J$. Finally, the backtracking procedure is applied to complete the construction.

\begin{tabular}{|c|cccc|cccc|cccc|cccc|}\hline
$g\backslash f$&\multicolumn{4}{c|}{1}&\multicolumn{4}{c|}{2}&\multicolumn{4}{c|}{3}&\multicolumn{4}{c|}{4}\\
\hline
0&5     &448   &0     &-     &13    &236   &0     &-     &4     &16    &0     &-     &13    &152   &64    &-     \\
1&16    &2168  &0     &-     &152   &5652  &136   &0     &77    &1004  &0     &-     &181   &1770  &70    &0     \\
2&0     &0     &0     &-     &43    &0     &0     &-     &50    &512   &0     &-     &76    &63    &0     &-     \\
3&0     &0     &0     &-     &105   &0     &0     &-     &104   &0     &0     &-     &210   &0     &0     &-     \\
4&0     &0     &0     &-     &802   &1660  &0     &-     &958   &7136  &0     &-     &2604  &4685  &0     &-     \\
5&0     &0     &0     &-     &272   &0     &0     &-     &584   &4448  &0     &-     &1856  &551   &0     &-     \\
\hline
$g\backslash f$&\multicolumn{4}{c|}{5}&\multicolumn{4}{c|}{6}&\multicolumn{4}{c|}{7}&\multicolumn{4}{c|}{8}\\
\hline
0&5     &178   &240   &-     &5     &24    &0     &-     &17    &48    &0     &-     &12    &62    &0     &-     \\
1&22    &112   &0     &-     &131 &1416  &0     &-     &420   &1951  &8     &0     &397   &1786  &70    &-     \\
2&9     &7     &0     &-     &37    &0     &0     &-     &436   &488   &10    &0     &824   &803   &40    &\textbf{12}    \\
3&2     &0     &0     &-     &102   &0     &0     &-     &1166  &0     &0     &-     &1822  &0     &0     &-     \\
4&131   &0     &0     &-     &2132  &154   &0     &-     &29290 &2537  &0     &-     &47184 &12699 &0     &-     \\
5&62    &11    &0     &-     &1122  &0     &0     &-     &50988 &11861 &0     &-     &89576 &61126 &57    &-     \\
\hline
\end{tabular}

An extension $G$ of $F$ is obtained by adding some coordinate functions, in
that case $\comp(F)$ becomes a subspace of components space of $G$.

\begin{lemma}\label{DELTA}
If a $(m,n)$-function $F$ has an \APN extension then $\Delta_F \leq 2^{m-n+1}$.
\end{lemma}

The vectorial $(m,m-2)$-function $F$ has an \APN extension , if and only if,
for all $(x,y,z,t)\in Q_F$ the system \rouge{of} quadratic equations:
\begin{equation}
\label{SYSTEM}
               g(x) + g(y) + g(z) + g(t) \not= 0 \Leftrightarrow  \big(g(x) + g(y) + g(z) + g(t)\big)^3 = 1.
\end{equation}
is solvable in $\fd(4)$. We remark that   $\big(g(x) + g(y) + g(z) + g(t)\big)^3 $ equal to:
\begin{equation*}
          x^3 + y^3 + z^3 + t^3 + xy(x+y) + xz(x+z)+ xt(x+t) + yz(y+z) + yt(y+t)+ zt(z+t).
\end{equation*}
so we can transform system (\ref{SYSTEM}) in an affine system $S_F$ with $N$ equations and $q(q+1)/2$ unknowns,
introducing  $q$ Boolean variables $x^3$, and   $q(q-1)/2$ variables $xy(x+y)$, where $q=2^m$.

\begin{lemma}\label{XYZT}
        If the affine system $S_F$ has no solution then $F$ has no \APN extension.
\end{lemma}

We say that an $(m,m-2)$-vectorial function passes the extension test if it satisfies conditions of Lemma \ref{DELTA} and Lemma \ref{XYZT}.  Even it is an hard task, it is possible to use 
the following procedure to "classify" all \APN functions 
of type $(\alpha,\beta)$ that are quartic extensions of 
a $(6,3)$-vectorial cubic. Let $\mathcal E$ be a set of  $(m,n)$-functions. We define
$\ext( \mathcal E )$ as the set of extensions $(F,f)$
having $(\alpha,\beta)$ type that satisfy Lemma \ref{DELTA} and 
"filtered" by invariant $\invec$.  Starting from $\mathcal E_0 :=\{ h\}$
where $\deg(h)=4$ and $\kappa(f)=\alpha$, 
we contruct $\mathcal E_1 = \ext( \mathcal E_0 )$,  $\mathcal E_2 = \ext( \mathcal E_1 )$,
and  $\mathcal E_3 = \ext( \mathcal E_2 )$. We keep the $(6,4)$-function passing
the extension test, and we terminate by a backtracking algorithm to identify
\APN extension, and then 2-level \APN functions.

Applying the procedure using ? 
for the pair (1.75,4), \textcolor{black}{there are 8 quartic classes to initialise the construction process, 4 of which produce \APN functions.} \bleu{The 506880 \APN functions obtained 
after the backtracking phase are not necessarily at 2-spectral levels, but all 16384 functions of type (1.75,4) are ultimately \CCZ-equivalent 
to Dublin permutation.}

\begin{fact} 
    The classes of Dublin classes are uniq among 
    bihomogeneous two level maps of type $(1.75,4)$.
\end{fact}

\section{extension of bent}

By  \cite{PolujanThesis}, we know there are 13 class of $(6,3)$-bent
functions. The invariant $J$ distinguish 9 of the 13 class of bent. 
One obtain a discriminate invariant adding a invariant $K$ defined as
follows. 

For an $(m,n)$-vectorial function $F$, we consider the mapping 
from $\bst(s,t,m)$ to $\vst(m,n,s+2,t)$ that sends 
$f_{\rm c}(k,m)$ to $f_F$. The dimension of the kernel 
of this mapping is an invariant $K$, which takes six values 
on $(6,3)$-bent functions. This invariant can be combined with 
$J$ to obtain a discriminating invariant for bent functions. 
Incidentally, the same invariant takes 715 values over the set of all known \APN functions.

Each support of an \APN function contains $1{,}395$ three-dimensional subspaces.
When considering all 3-dimensional subspaces of the known \APN functions, 
we observe that $3{,}455$ of these subspaces are bent spaces.

\begin{table}[htbp]
\caption{\label{BENT}number of bent 3-space in the component of an \APN }
\begin{tabular}{|l|ccccccccccc|}
\hline
     $\sharp$-bent  &0   &8 &16 &21 &24 &44 &48 &60 &74 &75 &140\\
     \hline
     mult. &504 &99 &78 &1 &21 &2 &2 &2 &4 &2 &1 \\
     \hline
\end{tabular}
\end{table}

A more detailed examination shows that these subspaces rely on five bent classes. 
Consequently, 8 of the 14 $(6,3)$-bent classes do not appear in the supports 
of the known \APN functions. It is therefore important to clarify this situation.
The backtracking algorithm we have described is sufficiently efficient 
to compute all \APN extensions of the 13 bent classes. As a first step, we compute all $3 \times 8$ echelon matrices; there are
\[
\binogauss(8,3) + \binogauss(8,3)^2 + 1 = 108{,}206
\]
such matrices. To decide whether a \APN-extension of a $(6,3)$-bent 
function exists, one can independently run $108{,}206$ instances of 
the backtracking procedure. The computation confirm that 8 class of 
$(6,3)$-bent functions are not \APN-expandable.  The remaining 5 classes 
yield 150 EA-classes (without new \CCZ-class), 
which are distributed as follows:

\begin{table}[h!]
\caption{\label{APNEXT}APN extension of bent class : 150 class}
\begin{tabular}{|l|ccccc|}
\hline
     code  &1 &2 &3 &5 &6\\
     \hline
    $\sharp$-soluce  &13441  & 58426   &61720   &103  &776 \\
    $\sharp$-class &3      &    12   & 13     &54   &111 \\
     \hline
\end{tabular}
\end{table}

\begin{fact}
        All the \APN extension 
        of (6,3)-bent functions 
        are known. 
\end{fact}

\section{Conclusion}

Our study represents an important step in understanding 6-bit APN functions, providing strong support for the completeness of the 14 known CCZ-classes. However, during our work, we realized that many possibilities still remain to be explored to fully resolve the 
case of two level situation.

\nocite{*}
\bibliographystyle{plain} 
\bibliography{switching.bib}

\end{document}